\documentclass[11pt]{article}
\usepackage{tikz}
\usepackage{subcaption}
\usepackage{amsthm}
\usepackage{amssymb}
\usepackage{amsmath}
\usepackage{graphicx}
\usepackage[authoryear]{natbib}
\newtheorem*{theorem*}{Theorem}
\newtheorem{theorem}{Theorem}
\newtheorem{lemma}{Lemma}
\newtheorem{cor}{Corollary}

\usetikzlibrary{positioning,chains,fit,shapes,calc}

\begin{document}
% \begin{frontmatter}
	
	\title{Saturating Stable Matchings}
	
	\author{Muhammad Maaz \\ \texttt{m.maaz@mail.utoronto.ca} \\ University of Toronto\footnote{1 King's College Circle, Toronto, ON, Canada M5S 1A8}}
	%\author{Muhammad Maaz}
	%\address{University of Toronto}
	%\ead{m.maaz@mail.utoronto.ca}
	\date{}

\maketitle

\begin{abstract}
I relate bipartite graph matchings to stable matchings. I prove a necessary and sufficient condition for the existence of a saturating stable matching, where every agent on one side is matched, for all possible preferences. I extend my analysis to perfect stable matchings, where every agent on both sides is matched.
\end{abstract}

% \end{frontmatter}

\section{Introduction}
A bipartite graph is a graph $G$ with two disjoint vertex sets, $X$ and $Y$, and an edge set $E$, such that there is no edge connecting two vertices in the same set. A common goal in bipartite graphs is to connect these two sets in a \emph{matching}, defined as a subset of $E$ such that no two edges share a vertex. If a vertex is the endpoint of an edge in $M$, we say it is \emph{matched}; otherwise it is unmatched. \citep{westgraphtheory}

\citet{hall1935} gives a necessary and sufficient condition for a bipartite graph to have an \emph{$X$-saturating matching}, where every vertex $x \in X$ is matched. 

When we imagine vertices as agents and allow them to have preferences over the other side, we have \citet{gs}'s classic stable marriage problem. A matching is stable if there does not exist a vertex pair $(x,y) \in X \times Y$ which are not matched together but prefer each other to their partner (note that their partner may be no one - i.e. they are unmatched) under that matching; we call this a \emph{blocking pair}. 

\citet{gs} prove there always exists a stable matching, giving a constructive proof by developing their famed deferred acceptance algorithm. Since their paper, the study of matchings has been taken on by economists in a rather different way than they are studied in graph theory, focusing on stability rather than combinatorics \citep{rothsoto}.

So, we have two major classical theorems looking at matchings from two different perspectives. \citet{hall1935} gives a theorem for the existence of a saturating matching (not necessarily stable), while \citet{gs} give a theorem for the existence of a stable matching (not necessarily saturating).

As matching applications proliferate, it may be desirable from a social perspective to have every agent matched (e.g., match everyone to a vaccine or match every medical student to a residency). Hence the question arises: can we find \emph{saturating stable} matchings?

\section{Main result}

Consider a bipartite graph $G = (X+Y, E)$. The neighborhood $N(x)$ of a vertex $x$ is the set of vertices adjacent to, or sharing an edge with, $x$, and the neighborhood of a set of vertices is the union of each vertex's neighborhood. The degree $deg(x)$ is the number of adjacent vertices, or $|N(x)|$.

Define $P$ to be a set of preference relations, the elements of which are each vertex's strict preference relation over the vertices in its neighborhood (it's \emph{acceptable} partners). If $x_1$ prefers $y_1$ to $y_2$, we write $x_1: y_1 \succ y_2$. Every vertex prefers an acceptable vertex to being unmatched, and prefers being unmatched to an unacceptable vertex. The set of all possible preference instances is $\mathbb{P}$. SM means stable matching below.

Define the following two conditions for some vertex $x$:
\begin{equation}
|N(N(x))| \leq |N(x)|
\end{equation}
\begin{equation}
\exists y \in N(x) (deg(y) = 1)
\end{equation}

\begin{lemma}
\label{xmatched}
If some $x \in X$ satisfies condition (1) or (2) (or both), then it is matched in all SMs in all preference instances.
\end{lemma}
\begin{proof}
 For contradiction's sake, assume that this $x$ is unmatched in some SM. By assumption, this $x$ satisfies $|N(N(x))| \leq |N(x)|$ or $\exists y \in N(x) (deg(y)=1)$, or both.
\begin{quote}
\begin{description}
\item[Case 1: $\exists y \in N(x) (deg(y)=1)$] Observe that this $y$ is unmatched, as it only has one acceptable vertex, and that is $x$, which is unmatched. This $x$ and $y$ form a blocking pair so this matching is unstable, hence contradiction.
\item[Case 2: $\nexists y \in N(x) (deg(y)=1)$] $x$ must satisfy $|N(N(x))| \leq |N(x)|$. Observe that $x$ is only unmatched if all of the vertices in $N(x)$ are already matched. Since vertices $N(x)$ can only be matched to vertices in $N(N(x))$, this means that $|N(N(x))|-1$ (the number of vertices in $N(N(x))$ excluding $x$ itself) vertices are matched to $|N(x)|$ vertices. By the pigeonhole principle, this is a contradiction.
\end{description}
\end{quote}
This exhausts all the cases.
\end{proof}

\begin{lemma}
\label{xunmatched}
If some $x \in X$ does not satisfy condition (1) nor (2), then there exists some preference instance $P \in \mathbb{P}$ under which it is unmatched in every SM.
\end{lemma}
\begin{proof}
Observe that this $x$ is such that $|N(N(x))| > |N(x)|$ and $(\forall y \in N(x))(deg(y)>1)$.\footnotemark

\footnotetext{Note that the negation of $deg(y)=1$ is $deg(y)>1$, because it is in some vertex's neighborhood, and so $deg(y)>0$.}

By the assumption, there exists some $x \in X$ that satisfies $|N(N(x))| > |N(x)|$ and $\forall y \in N(x) (deg(y)>1)$. It suffices to show there exists a $P$ which leaves $x$ unmatched.

Consider the following preference instance $P$:
\begin{itemize}
\item $\forall a \in N(x): i \succ x$, where $i \in N(N(x))-x$
\item $\forall b \in N(N(x))-x: j \succ k$, where $j \in N(x)$ and $k \in N(b)-N(x)$
\item all other preference relations are allowed to vary
\end{itemize}

Roughly speaking, $P$ states that all of $x$'s options prefer all their other acceptable vertices to $x$ itself, and that all of $x$'s competitors prefer to be matched to a vertex in $N(x)$ over any other vertex that they find acceptable.

Now we can show that $x$ is unmatched in all SMs. Assume for contradiction's sake that there is some SM $M$ in which $x$ is matched to some vertex in $N(x)$. Observe that $\forall v \in N(N(x))-x$ are matched to a vertex in $N(x)$. If some $v$ wasn't, then, due to the stated preferences, and the fact that all vertices in $N(x)$ have $deg > 1$ by the second part of the assumption, said $v$ would form a blocking pair with some vertex in $N(x)$.

However, this means that $|N(N(x))|$ vertices (all of $x$'s competitors plus $x$ itself) are matched to $|N(x)|$ vertices. But, since $|N(N(x))|>|N(x)|$, by the pigeonhole principle, this is a contradiction.

Therefore, $x$ is unmatched in all SMs under $P$. 
\end{proof}

Now the main result.

\begin{theorem}
\label{main}
Every SM is $X$-saturating for all preference instances if and only if for all $x \in X$, conditions (1) or (2) (or both) hold.
\end{theorem}
\begin{proof}
First, the if direction. For all $x \in X$, conditions (1) or (2) or both hold. By Lemma \ref{xmatched}, the desired statement holds.

Next, the only if direction. The contrapositive, where some $x$ satisfies neither condition, holds by Lemma \ref{xunmatched}.

\end{proof}

\subsection{Equivalent statements}

For the market designer, it is not necessarily important if \emph{all} SMs are saturating, but if at least one \emph{exists}. Further, in a real-world matching market that uses the Gale-Shapley algorithm, a very particular SM is yielded, which is either the $X$-optimal or $Y$-optimal SM (depending on the algorithm configuration)\footnotemark, so the designer may be particularly interested in whether the outputted SM is saturating.

\footnotetext{The $X$-optimal SM is such that $\forall x \in X$ prefers it to every other SM. The $X$-pessimal SM is such that $\forall x \in X$ prefers every other SM to it. The $X$-pessimal SM is also the $Y$-optimal SM \citep{rothsoto}.}

Interestingly, these are really all the same question, thanks to the following.

\begin{theorem*}[\citet{mcvitiewilson1970}]
	``In a marriage problem of n men and k women if any person is unmarried in one stable marriage solution he or she will be unmarried in all the stable solutions."
\end{theorem*} 

So, if a single SM is $X$-saturating (no one is ``unmarried"), then any other SM is also $X$-saturating (including the $X$-optimal one, and the $X$-pessimal one), and indeed all of them. 

\begin{lemma}
	\label{equiv}
	For a given preference instance $P$, an arbitrary SM is $X$-saturating if and only if all SMs are $X$-saturating.
\end{lemma}
\begin{proof}
	Follows from \citet{mcvitiewilson1970}.
\end{proof}

\begin{cor}
	For a given preference instance $P$, let the set of all SMs be $\mathbb{M}$. Then, for all preference instances, an arbitrary $M \in \mathbb{M}$ is $X$-saturating if and only if for all $x \in X$, at least one of conditions (1) and (2) hold.
\end{cor}
\begin{proof}
	Follows from Theorem \ref{main} and Lemma \ref{equiv}.
\end{proof}

Thus, the biconditional in Theorem \ref{main} is the same for the existence of an $X$-saturating matching, the $X$-optimal SM being $X$-saturating, or for \emph{any} arbitrary SM the market designer is interested in. This equivalence holds for subsequent results in this paper, per Lemma \ref{equiv}.

\section{Applications}

\subsection{Demonstrative examples}

\begin{figure}
	\centering
	\caption{Examples of bipartite graphs}
	\label{ex}
	     \begin{subfigure}[b]{0.20\textwidth}
          \centering
          \resizebox{\linewidth}{!}{
	\begin{tikzpicture}[every node/.style={circle}]
	
	\begin{scope}[start chain=going below,node distance=7mm]
	\node[draw, on chain] (x1) [label=left: $x_1$] {};
	\node[draw, on chain] (x2) [label=left: $x_2$] {};
	\end{scope}
	
	\begin{scope}[xshift=2cm, start chain=going below,node distance=7mm]
	\node[draw, on chain] (y1) [label=right: $y_1$] {};
	\node[draw, on chain] (y2) [label=right: $y_2$] {};
	\end{scope}
	
	\draw (x1) -- (y1);
	\draw (x2) -- (y2);
	\draw (x1) -- (y2);
	
	\node [above=of x1,label=below:$X$] {};
	\node [above=of y1,label=below:$Y$] {};
	
	\end{tikzpicture}
	}  
          \caption{}
     \end{subfigure} %
     \begin{subfigure}[b]{0.20\textwidth}
          \centering
          \resizebox{\linewidth}{!}{
          
          \begin{tikzpicture}[every node/.style={circle}]
	
	\begin{scope}[start chain=going below,node distance=7mm]
	\node[draw, on chain] (x1) [label=left: $x_1$] {};
	\node[draw, on chain] (x2) [label=left: $x_2$] {};
	\end{scope}
	
	\begin{scope}[xshift=2cm, start chain=going below,node distance=7mm]
	\node[draw, on chain] (y1) [label=right: $y_1$] {};
	\node[draw, on chain] (y2) [label=right: $y_2$] {};
	\node[draw, on chain] (y3) [label=right: $y_3$] {};
	\end{scope}
	
	\draw (x1) -- (y1);
	\draw (x2) -- (y2);
	\draw (x1) -- (y2);
	\draw (x2) -- (y3);
	
	\node [above=of x1,label=below:$X$] {};
	\node [above=of y1,label=below:$Y$] {};
	
	\end{tikzpicture}}  
          \caption{}
     \end{subfigure}
     \begin{subfigure}[b]{0.20\textwidth}
          \centering
          \resizebox{\linewidth}{!}{
          
          \begin{tikzpicture}[every node/.style={circle}]
	
	\begin{scope}[start chain=going below,node distance=7mm]
	\node[draw, on chain] (x1) [label=left: $x_1$] {};
	\node[draw, on chain] (x2) [label=left: $x_2$] {};
	\node[draw, on chain] (x3) [label=left: $x_3$] {};
	\end{scope}
	
	\begin{scope}[xshift=2cm, start chain=going below,node distance=7mm]
	\node[draw, on chain] (y1) [label=right: $y_1$] {};
	\node[draw, on chain] (y2) [label=right: $y_2$] {};
	\node[draw, on chain] (y3) [label=right: $y_3$] {};
	\node[draw, on chain] (y4) [label=right: $y_4$] {};
	\end{scope}
	
	\draw (x1) -- (y1);
	\draw (x2) -- (y2);
	\draw (x1) -- (y2);
	\draw (x3) -- (y2);
	\draw (x3) -- (y3);
	\draw (x1) -- (y4);
	\draw (x2) -- (y4);
	\draw (x2) -- (y1);
	
	\node [above=of x1,label=below:$X$] {};
	\node [above=of y1,label=below:$Y$] {};
	
	\end{tikzpicture}}  
          \caption{}
     \end{subfigure}
\end{figure}
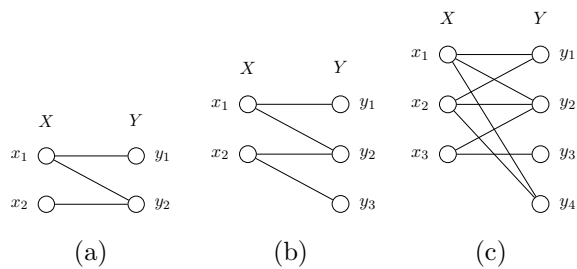

Looking at Figure \ref{ex}a, conditions (1) and (2) are violated for $x_2$, since $|N(N(x_2))|=2 > |N(x_2)|=1$, so $X$-saturating SMs do not exist for all preference instances.

In Figure \ref{ex}b, which simply added one vertex to \ref{ex}a, both conditions now hold for $x_2$.

Lastly, in Figure \ref{ex}c, condition (1) holds for $x_1$ and $x_2$. While $x_3$ violates condition (1), it does satisfy (2), as $deg(y_3)=1$ and $y_3 \in N(x_3)$.

\subsection{Perfect matchings}

\citet{gs} considered $|X|=|Y|=n$ and every $x \in X$ is acceptable to $y \in Y$ (and vice versa) (i.e. preferences are complete). In graph theory, this is a \emph{complete bipartite graph}. \citet{gs} say that no vertex is unmatched after the execution of their algorithm. In graph-theoretic terms, for all preference instances, the SM given by their algorithm is \emph{perfect}, meaning that every vertex is matched. 

This is actually implied by Theorem \ref{main}, as $\forall x \in X$ and $\forall y \in Y$ satisfy condition (1): $\forall x \in X$, $|N(x)|=n$ and $|N(N(x))|=n$ (due to being a complete bipartite graph) so condition (1) is fulfilled, and similarly $\forall y \in Y$. Therefore, by Theorem \ref{main}, for all preference instances, every SM is $X$-saturating and $Y$-saturating, and hence perfect.

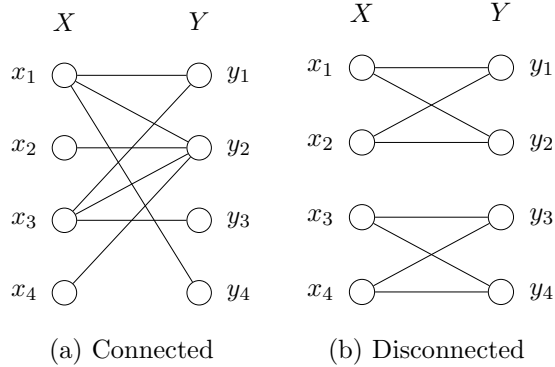
\begin{figure}
\centering
\caption{A connected and a disconnected matching market}
\label{pcs}
     \begin{subfigure}[b]{0.30\textwidth}
          \centering
          \resizebox{\linewidth}{!}{
	\begin{tikzpicture}[every node/.style={circle}]
	
	\begin{scope}[start chain=going below,node distance=7mm]
	\node[draw, on chain] (x1) [label=left: $x_1$] {};
	\node[draw, on chain] (x2) [label=left: $x_2$] {};
	\node[draw, on chain] (x3) [label=left: $x_3$] {};
	\node[draw, on chain] (x4) [label=left: $x_4$] {};
	\end{scope}
	
	\begin{scope}[xshift=2cm, start chain=going below,node distance=7mm]
	\node[draw, on chain] (y1) [label=right: $y_1$] {};
	\node[draw, on chain] (y2) [label=right: $y_2$] {};
	\node[draw, on chain] (y3) [label=right: $y_3$] {};
	\node[draw, on chain] (y4) [label=right: $y_4$] {};
	\end{scope}
	
	\draw (x1) -- (y1);
	\draw (x2) -- (y2);
	\draw (x1) -- (y2);
	\draw (x3) -- (y1);
	\draw (x4) -- (y2);
	\draw (x3) -- (y2);
	\draw (x1) -- (y4);
	\draw (x3) -- (y3);
	
	\node [above=of x1,label=below:$X$] {};
	\node [above=of y1,label=below:$Y$] {};
	
	\end{tikzpicture}
	}  
          \caption{Connected}
     \end{subfigure} %
     \begin{subfigure}[b]{0.30\textwidth}
          \centering
          \resizebox{\linewidth}{!}{
          
          \begin{tikzpicture}[every node/.style={circle}]
	
	\begin{scope}[start chain=going below,node distance=7mm]
	\node[draw, on chain] (x1) [label=left: $x_1$] {};
	\node[draw, on chain] (x2) [label=left: $x_2$] {};
	\node[draw, on chain] (x3) [label=left: $x_3$] {};
	\node[draw, on chain] (x4) [label=left: $x_4$] {};
	\end{scope}
	
	\begin{scope}[xshift=2cm, start chain=going below,node distance=7mm]
	\node[draw, on chain] (y1) [label=right: $y_1$] {};
	\node[draw, on chain] (y2) [label=right: $y_2$] {};
	\node[draw, on chain] (y3) [label=right: $y_3$] {};
	\node[draw, on chain] (y4) [label=right: $y_4$] {};
	\end{scope}
	
	\draw (x1) -- (y1);
	\draw (x2) -- (y2);
	\draw (x1) -- (y2);
	\draw (x2) -- (y1);
	
	\draw (x3) -- (y3);
	\draw (x3) -- (y4);
	\draw (x4) -- (y3);
	\draw (x4) -- (y4);
	
	\node [above=of x1,label=below:$X$] {};
	\node [above=of y1,label=below:$Y$] {};
	
	\end{tikzpicture}}  
          \caption{Disconnected}
     \end{subfigure}
\end{figure}

In fact, if our matching market is in one ``piece'', then the only way to obtain a perfect matching if $|X|=|Y|=n$ is if preferences are complete. In Figure \ref{pcs}b, we can visually see that there are two different pieces - called \emph{components} in graph theory. Even though preferences are incomplete (e.g., $x_3$ does not find $y_1$ acceptable), clearly a perfect SM will always exist. 

I first consider graphs with only one component, like Figure \ref{pcs}a. A graph with only one component is called a \emph{connected graph}, defined by a path existing between any two vertices \citep{westgraphtheory}.

\begin{theorem}
\label{perfmatch}
Given a connected bipartite graph $G=(X+Y,E)$ with $|X|=|Y|=n$, all SMs are perfect for all preference instances if and only if $G$ is a complete bipartite graph.
\end{theorem}
\begin{proof}
First, the if direction. If $G$ is a complete bipartite graph, then $\forall x \in X$ and $\forall y \in Y$ satisfy condition (1), as discussed above, so by Theorem \ref{main}, all SMs in all preference instances are $X$-saturating and $Y$-saturating, meaning perfect.

In the other direction, proceed by induction on $n$. For the base case $n=1$, there is one vertex each in $X$ and $Y$, and they have an edge as $G$ is connected. Clearly, this is a complete bipartite graph.

Next, assume that for some $n=k$, if all SMs are perfect for all preference instances, then $G_k$ is a complete bipartite graph.

We wish to show that for $n=k+1$, $G_{k+1}$ is also a complete bipartite graph. $G_{k+1}$ is formed by adding a vertex to each $X$ and $Y$, called $x$ and $y$ respectively. Because $G_{k+1}$ is connected, at least one of $x$ or $y$ must be connected to a vertex other than $y$ or $x$ respectively. Without loss of generality, say it is $x$, which is connected to some not-$y$ vertex $v \in Y$. 

Observe that $|N(v)|=k+1$, and so $|N(N(x))|=k+1$. $x$ must be matched in all SMs in all preference instances. By the contrapositive of Lemma \ref{xunmatched}, $x$ must satisfy $k+1 \geq|N(x)| \geq |N(N(x))| = k+1$, and so $|N(x)|=k+1$, which means $x$ is connected to every vertex in $Y$.

By similar reasoning, $y$ is also connected to every vertex in $X$, and hence we have a complete bipartite graph.

Thus, by induction, the ``only if'' statement holds.
\end{proof}

I now extend Theorem \ref{perfmatch} for the case of all graphs, not just connected ones. Looking at Figure \ref{pcs}b, the two components individually exhibit complete preferences over vertices in the same component. There is a special name for such components: these are called \emph{bicliques} \citep{westgraphtheory}. In a biclique, every vertex is connected to every vertex in the other set (a generalization of cliques to bipartite graphs). A complete bipartite graph is itself a biclique.

\begin{cor}
Given a bipartite graph $G=(X+Y,E)$ with $|X|=|Y|$, all SMs are perfect for all preference instances if and only if every component of $G$ is a biclique.
\end{cor}
\begin{proof}
Follows by applying Theorem \ref{perfmatch} to each component of the graph.
\end{proof}

\subsection{Matching with compatibility constraints}

\citet{maazpapa2020} developed the matching with compatibility constraints problem by studying the Canadian medical residency match and point out that positions are designated as either for English speakers or French speakers; some students, being bilingual, can apply to either. Theorem \ref{main} allows us to easily generalize their model to $n$ compatibility classes.

The vertex set $X$ is divided into $n$ possibly overlapping sets $X=A_1 \cup A_2 \cup A_3 ... \cup A_n$. The set $A_i - \cup_{j \neq i} A_j$ must be nonempty for all $i,j \in [1,n]$, meaning that there must be at least one vertex in each class that is not in any other class\footnotemark. The set $Y$ is partitioned into $n$ disjoint subsets $Y=B_1 \cup B_2 ... \cup B_n$. See Figure \ref{compconstraints} for a schematic. A vertex can not find another vertex acceptable if they do not belong to the same class; otherwise they may, but not necessarily. In the special case of compatibility-wise complete (or CW-complete) preferences that \citet{maazpapa2020} study, every vertex finds every vertex in the same class acceptable.

\footnotetext{This is a generalization of \citet{maazpapa2020}'s restriction that there must be a non-zero amount of students that speak only English and a non-zero amount that speak only French.}

\begin{theorem}
\label{compconstraints}
With $n$ compatibility classes, every SM is $X$-saturating in all instances of CW-complete preferences if and only if $|B_i| \geq |A_i|$ for all $1 \leq i \leq n$.
\end{theorem}
\begin{proof}
First, the if direction. Take an arbitrary $x \in X$. It belongs to one or more compatibility classes; let this list of classes be stored in the vector $q$. Then $|N(x)|= \sum_{i \in q} |B_i|$. Further, observe that $N(N(x))$ is the set of all vertices in $X$ that also belong to the same compatibility classes, including $x$ itself. Thus, $|N(N(x))| =  \sum_{i \in q} |A_i|$. Because $|B_i| \geq |A_i|$ for all $1 \leq i \leq n$, condition (1) holds for this $x$, and indeed for all $x$. By Lemma \ref{xmatched}, the result follows.

	Next, the ``only if'' direction. Assume for contradiction's sake that there exists a compatibility class $i$ such that $|B_i| < |A_i|$. There exists at least one vertex $x \in X$ that is in the $i$th compatibility class and not in any other class. Then, $|N(x)| = |B_i| < |A_i| = |N(N(x))|$, so it does not fulfill condition (1). And, it does not fulfill condition (2) either because it cannot be connected to a vertex $y$ with degree of 1, as $y$ would violate CW-complete preferences, unless $x$ is the only vertex, but that violates $|B_i| < |A_i|$. By Lemma \ref{xunmatched}, there exists a preference instance with a SM that is not $X$-saturating, which is a contradiction. 
\end{proof}

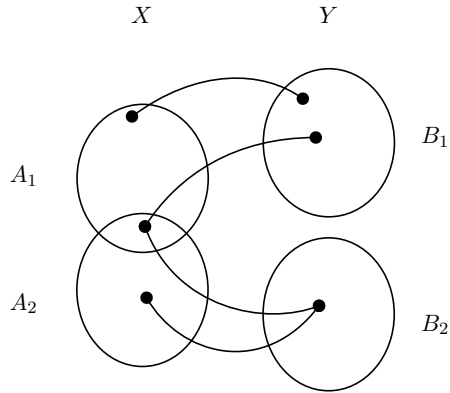
\begin{figure}
\centering
\caption{Matching with compatibility constraints with 2 and 3 classes. Lines indicate compatibility between every vertex in the two subsets touched by the line's endpoints; under CW-complete preferences, lines also indicate acceptability.}
\label{compconstraints}
\begin{subfigure}[b]{\textwidth}
          \centering
          \resizebox{0.5\linewidth}{!}{
          		\tikzset{every picture/.style={line width=0.75pt}} %set default line width to 0.75pt        

\begin{tikzpicture}[x=0.75pt,y=0.75pt,yscale=-1,xscale=1]
%uncomment if require: \path (0,342); %set diagram left start at 0, and has height of 342

%Shape: Ellipse [id:dp9455051011392746] 
\draw   (230.66,81.5) .. controls (253.02,81.54) and (271.12,102.06) .. (271.07,127.33) .. controls (271.02,152.59) and (252.85,173.04) .. (230.48,173) .. controls (208.11,172.96) and (190.02,152.44) .. (190.07,127.17) .. controls (190.12,101.91) and (208.29,81.46) .. (230.66,81.5) -- cycle ;
%Shape: Ellipse [id:dp8907540153897295] 
\draw   (230.53,149) .. controls (252.89,149.04) and (270.99,170.23) .. (270.94,196.33) .. controls (270.89,222.42) and (252.71,243.54) .. (230.34,243.5) .. controls (207.98,243.46) and (189.89,222.27) .. (189.94,196.17) .. controls (189.99,170.08) and (208.16,148.96) .. (230.53,149) -- cycle ;
%Shape: Ellipse [id:dp5614325267743172] 
\draw   (345.66,59.5) .. controls (368.02,59.54) and (386.12,80.06) .. (386.07,105.33) .. controls (386.02,130.59) and (367.85,151.04) .. (345.48,151) .. controls (323.11,150.96) and (305.02,130.44) .. (305.07,105.17) .. controls (305.12,79.91) and (323.29,59.46) .. (345.66,59.5) -- cycle ;
%Shape: Ellipse [id:dp618530555463572] 
\draw   (345.53,164) .. controls (367.89,164.04) and (385.99,185.23) .. (385.94,211.33) .. controls (385.89,237.42) and (367.71,258.54) .. (345.34,258.5) .. controls (322.98,258.46) and (304.89,237.27) .. (304.94,211.17) .. controls (304.99,185.08) and (323.16,163.96) .. (345.53,164) -- cycle ;
%Curve Lines [id:da2910493106816676] 
\draw    (224,89) .. controls (264,59) and (305,60) .. (329.5,78) ;
\draw [shift={(329.5,78)}, rotate = 36.3] [color={rgb, 255:red, 0; green, 0; blue, 0 }  ][fill={rgb, 255:red, 0; green, 0; blue, 0 }  ][line width=0.75]      (0, 0) circle [x radius= 3.35, y radius= 3.35]   ;
\draw [shift={(224,89)}, rotate = 323.13] [color={rgb, 255:red, 0; green, 0; blue, 0 }  ][fill={rgb, 255:red, 0; green, 0; blue, 0 }  ][line width=0.75]      (0, 0) circle [x radius= 3.35, y radius= 3.35]   ;
%Curve Lines [id:da06853493659447984] 
\draw    (232,157) .. controls (252.5,125) and (287.5,102) .. (337.5,102) ;
\draw [shift={(337.5,102)}, rotate = 0] [color={rgb, 255:red, 0; green, 0; blue, 0 }  ][fill={rgb, 255:red, 0; green, 0; blue, 0 }  ][line width=0.75]      (0, 0) circle [x radius= 3.35, y radius= 3.35]   ;
\draw [shift={(232,157)}, rotate = 302.64] [color={rgb, 255:red, 0; green, 0; blue, 0 }  ][fill={rgb, 255:red, 0; green, 0; blue, 0 }  ][line width=0.75]      (0, 0) circle [x radius= 3.35, y radius= 3.35]   ;
%Curve Lines [id:da26743485766430375] 
\draw    (232,157) .. controls (250.5,208) and (306.5,219) .. (339.5,206) ;
\draw [shift={(339.5,206)}, rotate = 338.5] [color={rgb, 255:red, 0; green, 0; blue, 0 }  ][fill={rgb, 255:red, 0; green, 0; blue, 0 }  ][line width=0.75]      (0, 0) circle [x radius= 3.35, y radius= 3.35]   ;
\draw [shift={(232,157)}, rotate = 70.06] [color={rgb, 255:red, 0; green, 0; blue, 0 }  ][fill={rgb, 255:red, 0; green, 0; blue, 0 }  ][line width=0.75]      (0, 0) circle [x radius= 3.35, y radius= 3.35]   ;
%Curve Lines [id:da6169492310176543] 
\draw    (233,201) .. controls (253.5,237) and (305.5,251) .. (339.5,206) ;
\draw [shift={(339.5,206)}, rotate = 307.07] [color={rgb, 255:red, 0; green, 0; blue, 0 }  ][fill={rgb, 255:red, 0; green, 0; blue, 0 }  ][line width=0.75]      (0, 0) circle [x radius= 3.35, y radius= 3.35]   ;
\draw [shift={(233,201)}, rotate = 60.34] [color={rgb, 255:red, 0; green, 0; blue, 0 }  ][fill={rgb, 255:red, 0; green, 0; blue, 0 }  ][line width=0.75]      (0, 0) circle [x radius= 3.35, y radius= 3.35]   ;

% Text Node
\draw (147,118) node [anchor=north west][inner sep=0.75pt]   [align=left] {$\displaystyle A_{1}$};
% Text Node
\draw (147,197) node [anchor=north west][inner sep=0.75pt]   [align=left] {$\displaystyle A_{2}$};
% Text Node
\draw (401,94) node [anchor=north west][inner sep=0.75pt]   [align=left] {$\displaystyle B_{1}$};
% Text Node
\draw (401,210) node [anchor=north west][inner sep=0.75pt]   [align=left] {$\displaystyle B_{2}$};
% Text Node
\draw (222,20) node [anchor=north west][inner sep=0.75pt]   [align=left] {$\displaystyle X$};
% Text Node
\draw (338,20) node [anchor=north west][inner sep=0.75pt]   [align=left] {$\displaystyle Y$};

\end{tikzpicture}
	}
	\caption{2 compatibility classes}
\end{subfigure}
\begin{subfigure}[b]{\textwidth}
	\centering
	\resizebox{0.5\linewidth}{!}{
		\tikzset{every picture/.style={line width=0.75pt}} %set default line width to 0.75pt        

\begin{tikzpicture}[x=0.75pt,y=0.75pt,yscale=-1,xscale=1]
%uncomment if require: \path (0,342); %set diagram left start at 0, and has height of 342

%Shape: Ellipse [id:dp9455051011392746] 
\draw   (319.67,107) .. controls (342.04,107.04) and (360.13,129.46) .. (360.08,157.08) .. controls (360.02,184.69) and (341.85,207.04) .. (319.48,207) .. controls (297.11,206.96) and (279.02,184.54) .. (279.08,156.92) .. controls (279.13,129.31) and (297.3,106.96) .. (319.67,107) -- cycle ;
%Shape: Ellipse [id:dp8907540153897295] 
\draw   (290.81,170.99) .. controls (317.44,171.04) and (338.99,192.23) .. (338.94,218.33) .. controls (338.89,244.42) and (317.26,265.54) .. (290.63,265.49) .. controls (264,265.43) and (242.46,244.24) .. (242.51,218.14) .. controls (242.56,192.05) and (264.18,170.93) .. (290.81,170.99) -- cycle ;
%Shape: Ellipse [id:dp5614325267743172] 
\draw   (319.41,7.5) .. controls (341.09,7.54) and (358.63,26.7) .. (358.58,50.29) .. controls (358.54,73.88) and (340.93,92.97) .. (319.25,92.92) .. controls (297.57,92.88) and (280.03,73.73) .. (280.07,50.14) .. controls (280.12,26.55) and (297.73,7.46) .. (319.41,7.5) -- cycle ;
%Shape: Ellipse [id:dp9970314520144914] 
\draw   (346.06,171.09) .. controls (372.31,171.14) and (393.55,192.34) .. (393.5,218.43) .. controls (393.45,244.53) and (372.13,265.64) .. (345.88,265.59) .. controls (319.62,265.54) and (298.39,244.35) .. (298.44,218.25) .. controls (298.49,192.15) and (319.81,171.04) .. (346.06,171.09) -- cycle ;
%Shape: Ellipse [id:dp09730484726958899] 
\draw   (440.41,240.5) .. controls (462.09,240.54) and (479.63,259.7) .. (479.58,283.29) .. controls (479.54,306.88) and (461.93,325.97) .. (440.25,325.92) .. controls (418.57,325.88) and (401.03,306.73) .. (401.07,283.14) .. controls (401.12,259.55) and (418.73,240.46) .. (440.41,240.5) -- cycle ;
%Shape: Ellipse [id:dp3270169027631349] 
\draw   (201.41,241.5) .. controls (223.09,241.54) and (240.63,260.7) .. (240.58,284.29) .. controls (240.54,307.88) and (222.93,326.97) .. (201.25,326.92) .. controls (179.57,326.88) and (162.03,307.73) .. (162.07,284.14) .. controls (162.12,260.55) and (179.73,241.46) .. (201.41,241.5) -- cycle ;
%Curve Lines [id:da10776674692358945] 
\draw    (322,129) .. controls (354.5,104) and (344.5,72) .. (317.5,57) ;
\draw [shift={(317.5,57)}, rotate = 209.05] [color={rgb, 255:red, 0; green, 0; blue, 0 }  ][fill={rgb, 255:red, 0; green, 0; blue, 0 }  ][line width=0.75]      (0, 0) circle [x radius= 3.35, y radius= 3.35]   ;
\draw [shift={(322,129)}, rotate = 322.43] [color={rgb, 255:red, 0; green, 0; blue, 0 }  ][fill={rgb, 255:red, 0; green, 0; blue, 0 }  ][line width=0.75]      (0, 0) circle [x radius= 3.35, y radius= 3.35]   ;
%Curve Lines [id:da4300095012824605] 
\draw    (344,184) .. controls (376.5,159) and (387.5,67) .. (317.5,57) ;
\draw [shift={(317.5,57)}, rotate = 188.13] [color={rgb, 255:red, 0; green, 0; blue, 0 }  ][fill={rgb, 255:red, 0; green, 0; blue, 0 }  ][line width=0.75]      (0, 0) circle [x radius= 3.35, y radius= 3.35]   ;
\draw [shift={(344,184)}, rotate = 322.43] [color={rgb, 255:red, 0; green, 0; blue, 0 }  ][fill={rgb, 255:red, 0; green, 0; blue, 0 }  ][line width=0.75]      (0, 0) circle [x radius= 3.35, y radius= 3.35]   ;
%Curve Lines [id:da7637426615841207] 
\draw    (344,184) .. controls (376.5,159) and (458.5,183) .. (444.5,261) ;
\draw [shift={(444.5,261)}, rotate = 100.18] [color={rgb, 255:red, 0; green, 0; blue, 0 }  ][fill={rgb, 255:red, 0; green, 0; blue, 0 }  ][line width=0.75]      (0, 0) circle [x radius= 3.35, y radius= 3.35]   ;
\draw [shift={(344,184)}, rotate = 322.43] [color={rgb, 255:red, 0; green, 0; blue, 0 }  ][fill={rgb, 255:red, 0; green, 0; blue, 0 }  ][line width=0.75]      (0, 0) circle [x radius= 3.35, y radius= 3.35]   ;
%Curve Lines [id:da0874176249151386] 
\draw    (358,217) .. controls (375.5,198) and (426.5,196) .. (444.5,261) ;
\draw [shift={(444.5,261)}, rotate = 74.52] [color={rgb, 255:red, 0; green, 0; blue, 0 }  ][fill={rgb, 255:red, 0; green, 0; blue, 0 }  ][line width=0.75]      (0, 0) circle [x radius= 3.35, y radius= 3.35]   ;
\draw [shift={(358,217)}, rotate = 312.65] [color={rgb, 255:red, 0; green, 0; blue, 0 }  ][fill={rgb, 255:red, 0; green, 0; blue, 0 }  ][line width=0.75]      (0, 0) circle [x radius= 3.35, y radius= 3.35]   ;
%Curve Lines [id:da48867022436359164] 
\draw    (320,195) .. controls (335.5,246) and (394.5,270) .. (444.5,261) ;
\draw [shift={(444.5,261)}, rotate = 349.8] [color={rgb, 255:red, 0; green, 0; blue, 0 }  ][fill={rgb, 255:red, 0; green, 0; blue, 0 }  ][line width=0.75]      (0, 0) circle [x radius= 3.35, y radius= 3.35]   ;
\draw [shift={(320,195)}, rotate = 73.09] [color={rgb, 255:red, 0; green, 0; blue, 0 }  ][fill={rgb, 255:red, 0; green, 0; blue, 0 }  ][line width=0.75]      (0, 0) circle [x radius= 3.35, y radius= 3.35]   ;
%Curve Lines [id:da18218968747411024] 
\draw    (319,229) .. controls (334.5,280) and (396.5,290) .. (444.5,261) ;
\draw [shift={(444.5,261)}, rotate = 328.86] [color={rgb, 255:red, 0; green, 0; blue, 0 }  ][fill={rgb, 255:red, 0; green, 0; blue, 0 }  ][line width=0.75]      (0, 0) circle [x radius= 3.35, y radius= 3.35]   ;
\draw [shift={(319,229)}, rotate = 73.09] [color={rgb, 255:red, 0; green, 0; blue, 0 }  ][fill={rgb, 255:red, 0; green, 0; blue, 0 }  ][line width=0.75]      (0, 0) circle [x radius= 3.35, y radius= 3.35]   ;
%Curve Lines [id:da7171848796733766] 
\draw    (319,229) .. controls (306.5,288) and (267.5,299) .. (215.5,290) ;
\draw [shift={(215.5,290)}, rotate = 189.82] [color={rgb, 255:red, 0; green, 0; blue, 0 }  ][fill={rgb, 255:red, 0; green, 0; blue, 0 }  ][line width=0.75]      (0, 0) circle [x radius= 3.35, y radius= 3.35]   ;
\draw [shift={(319,229)}, rotate = 101.96] [color={rgb, 255:red, 0; green, 0; blue, 0 }  ][fill={rgb, 255:red, 0; green, 0; blue, 0 }  ][line width=0.75]      (0, 0) circle [x radius= 3.35, y radius= 3.35]   ;
%Curve Lines [id:da5556520726993677] 
\draw    (320,195) .. controls (248.5,205) and (219.5,242) .. (215.5,290) ;
\draw [shift={(215.5,290)}, rotate = 94.76] [color={rgb, 255:red, 0; green, 0; blue, 0 }  ][fill={rgb, 255:red, 0; green, 0; blue, 0 }  ][line width=0.75]      (0, 0) circle [x radius= 3.35, y radius= 3.35]   ;
\draw [shift={(320,195)}, rotate = 172.04] [color={rgb, 255:red, 0; green, 0; blue, 0 }  ][fill={rgb, 255:red, 0; green, 0; blue, 0 }  ][line width=0.75]      (0, 0) circle [x radius= 3.35, y radius= 3.35]   ;
%Curve Lines [id:da23319925327255153] 
\draw    (274,195) .. controls (230.5,191) and (193.5,243) .. (215.5,290) ;
\draw [shift={(215.5,290)}, rotate = 64.92] [color={rgb, 255:red, 0; green, 0; blue, 0 }  ][fill={rgb, 255:red, 0; green, 0; blue, 0 }  ][line width=0.75]      (0, 0) circle [x radius= 3.35, y radius= 3.35]   ;
\draw [shift={(274,195)}, rotate = 185.25] [color={rgb, 255:red, 0; green, 0; blue, 0 }  ][fill={rgb, 255:red, 0; green, 0; blue, 0 }  ][line width=0.75]      (0, 0) circle [x radius= 3.35, y radius= 3.35]   ;
%Curve Lines [id:da09580696906327613] 
\draw    (296.5,180) .. controls (195.5,174) and (193.5,243) .. (215.5,290) ;
\draw [shift={(215.5,290)}, rotate = 64.92] [color={rgb, 255:red, 0; green, 0; blue, 0 }  ][fill={rgb, 255:red, 0; green, 0; blue, 0 }  ][line width=0.75]      (0, 0) circle [x radius= 3.35, y radius= 3.35]   ;
\draw [shift={(296.5,180)}, rotate = 183.4] [color={rgb, 255:red, 0; green, 0; blue, 0 }  ][fill={rgb, 255:red, 0; green, 0; blue, 0 }  ][line width=0.75]      (0, 0) circle [x radius= 3.35, y radius= 3.35]   ;
%Curve Lines [id:da23351573934529046] 
\draw    (296.5,180) .. controls (236.5,175) and (210.5,87) .. (317.5,57) ;
\draw [shift={(317.5,57)}, rotate = 344.34] [color={rgb, 255:red, 0; green, 0; blue, 0 }  ][fill={rgb, 255:red, 0; green, 0; blue, 0 }  ][line width=0.75]      (0, 0) circle [x radius= 3.35, y radius= 3.35]   ;
\draw [shift={(296.5,180)}, rotate = 184.76] [color={rgb, 255:red, 0; green, 0; blue, 0 }  ][fill={rgb, 255:red, 0; green, 0; blue, 0 }  ][line width=0.75]      (0, 0) circle [x radius= 3.35, y radius= 3.35]   ;
%Curve Lines [id:da559286868165048] 
\draw    (320,195) .. controls (287.5,166) and (276.5,82) .. (317.5,57) ;
\draw [shift={(317.5,57)}, rotate = 328.63] [color={rgb, 255:red, 0; green, 0; blue, 0 }  ][fill={rgb, 255:red, 0; green, 0; blue, 0 }  ][line width=0.75]      (0, 0) circle [x radius= 3.35, y radius= 3.35]   ;
\draw [shift={(320,195)}, rotate = 221.74] [color={rgb, 255:red, 0; green, 0; blue, 0 }  ][fill={rgb, 255:red, 0; green, 0; blue, 0 }  ][line width=0.75]      (0, 0) circle [x radius= 3.35, y radius= 3.35]   ;

% Text Node
\draw (259,116) node [anchor=north west][inner sep=0.75pt]   [align=left] {$\displaystyle A_{1}$};
% Text Node
\draw (259,263) node [anchor=north west][inner sep=0.75pt]   [align=left] {$\displaystyle A_{3}$};
% Text Node
\draw (258,21) node [anchor=north west][inner sep=0.75pt]   [align=left] {$\displaystyle B_{1}$};
% Text Node
\draw (486,300) node [anchor=north west][inner sep=0.75pt]   [align=left] {$\displaystyle B_{2}$};
% Text Node
\draw (138,292) node [anchor=north west][inner sep=0.75pt]   [align=left] {$\displaystyle B_{3}$};
% Text Node
\draw (333,267) node [anchor=north west][inner sep=0.75pt]   [align=left] {$\displaystyle A_{2}$};

\end{tikzpicture}
	}
	\caption{3 compatibility classes}
\end{subfigure}
\end{figure}

\newpage
\bibliography{ref}

\end{document}